\documentclass[12pt]{article}
\usepackage{amscd,amsfonts,amsmath,amssymb,amstext,amsthm}
\title
{Selected results\\
on\\
Lie supergroups and their radial operators}
\author{Alan Huckleberry and Matthias Kalus\footnote
{Supported by SFB/TR 12,
\textit{Symmetry and universality in mesoscopic systems}, of
the Deutsche Forschungsgemeinschaft}}
\date{\today}
\theoremstyle{plain}
\newtheorem{theorem} {Theorem} [section]

\newtheorem{proposition}[theorem]{Proposition}

\theoremstyle{definition}

\begin{document}
\maketitle
\begin {abstract}
\noindent
Foundational material on complex Lie supergroups
and their radial operators is presented.  In particular,
Berezin's recursion formula for describing the radial parts
of fundamental operators in general linear
and ortho-symplectic cases is proved. Local versions of
results which are suitable for applications for computing characters
which are only defined on proper subdomains or covering spaces
thereof are established. 
\end {abstract}
\noindent
Using Howe-duality in a Fock space context one observes that certain integrals
of physical importance can be interpreted as numerical parts of
characters of holomorphic semigroup representations of pieces of Lie 
supergroups which appear as Howe partners of classical symmetry
(\cite{CFZ,HPZ}). The characters are
holomorphic superfunction of a parameter which is varying in a 
covering space of a domain in the base complex reductive 
group of the Lie supergroup at hand.  Such covering spaces, which
arise, e.g., due to the involvement of the metaplectic representation,
can be regarded as the domains of definition of the semigroup 
representations. They contain pieces of maximal tori so that at least in
a local sense one has the appropriate notions of radial functions
and operators. In fact the restrictions $\chi $ of the character
to this torus piece is the function, i.e., the integral, which is of interest.  
The character property implies that the integrals are eigenfunctions 
of the radial parts of Laplace-Casimir operators. In fact, in the cases
considered in \cite{CFZ,HPZ} the eigenvalues are zero! Thus in those
cases every element $D$ of the center universal enveloping algebra 
yields a differential equation which is of the form $\dot D(\chi )=0$
where $\dot D$ is the associated radial part of $D$. 

\bigskip\noindent
The above is conceptually pleasing, but to complete the task of obtaining
an explicit formula for the correlation functions $\chi $ one needs
more concrete information on the radial parts $\dot D$ of the
Laplace-Casimir operators.  Although all of the necessary information
is contained in the fundamental work of F. Berezin, gleaning it  
from (\cite{B}) and adapting it to the local setting indicated above
requires a serious effort. Originally the first author of the present
paper had hoped that it would be possible to provide at least a
roadmap of \cite{B} and include this in \cite{HPZ}. However, this 
proved to be an unsatisfactory solution, in particular due to the
additional length.  Furthermore, we felt that more than a roadmap
is needed.  Thus our work developed into the thesis project
\cite{Ka1}, the second author's paper \cite{Ka2} and our work
here.

\bigskip\noindent
Let us now briefly summarize this paper. The first section is
primarily devoted to an explanation of the method of Grassmann analytic
continuation and its application to defining a Lie supergroup associated
to a given Lie superalgebra $\mathfrak{g}=\mathfrak {g}_0+\mathfrak {g}_1$.  
We do discuss Berezin's definitions of the morphisms of multiplication, inverse
and the unit. However, since we are primarily interested in 
the differential operators, we pay most attention to the representations
of $\mathfrak g$ as left- and right-invariant derivations on the structure
sheaf.

\bigskip\noindent
The remainder of the paper is devoted to considerations of radial
operators, primarily in the cases of $\mathfrak {gl}$ and $\mathfrak {osp}$
which are of interest for the above mentioned Fock space applications.
The first general goal is to describe the radial part of a Laplace-Casimir
operator $D$ by $\dot D=J^{-1}P_DJ$ where $P_D$ is a constant coefficient
polynomial differential operator on the given maximal torus of the
Lie group $G$ associated to $\mathfrak {g}_0$.  A number of assumptions
are needed for this, in particular that the function $J$ should be an 
eigenfunction of the second order Laplace-Casimir operator which 
is defined in the usual way by an invariant nondegenerate supersymmetric 
bilinear form. This function appears as the square-root of 
the superdeterminant of a Jacobian of a coordinate chart which
identifies a neighborhood $\mathcal A$
of a regular point in a maximal torus with a product 
$\mathcal A \times \mathcal B$ where $\mathcal B$ is a subsupermanifold
which is defined as a local orbit of the Lie supergroup by conjugation 
(see the Appendix).

\bigskip\noindent
These assumptions are satisfied in the two cases of interest mentioned 
above. Furthermore, in both cases there is an important infinite series
$\{F_\ell\}$ of elements of the center of the universal enveloping algebra
which defines a series of $\{D_\ell\}$ Laplace-Casimir operators for which
the constant coefficient operators $P_{D_{\ell}}$ can be
described via a certain recursive procedure.  For example, in the
case of $\mathfrak {osp}$, if we use the standard basis for
the standard Cartan algebra, then the polynomials $F_\ell$
are defined by $F_\ell = \sum \varphi _i^{2\ell}+(-1)^\ell\sum \phi_j^{2\ell}$. 
Following Berezin we write $\dot D_\ell=J^{-1}T(F_\ell)J$. The goal is then
to understand the map $F_\ell \mapsto T(F_\ell)$.  Identifying the 
polynomials $F_\ell$ with the constant coefficient operators which they
define, the main result of Berizin is that 
$T(F_\ell)=F_\ell+Q_\ell$ where $Q_\ell$ is a polynomial in
$F_1, \ldots ,F_{\ell-1}$.   The proof of this result is discussed here
in substantial detail in $\S2$.

\bigskip\noindent
The last paragraphs of $\S2$ are devoted to local versions of
the formula mentioned above.  First, we show that the global
results of Berezin apply to give the same results on the
local product neighborhoods $\mathcal A\times \mathcal B$.
These then lift to the covering spaces mentioned at the outset
to give global results there by applying the identity principle.

\bigskip\noindent
Finally, in the appendix we outline the proof of 
$\dot D=J^{-1}P_DJ$ on $\mathcal A\times \mathcal B$ which 
is given in detail in (\cite{Ka1}, Chapter 4) by methods which 
are similar to those used by Helgason in the classical case.
Of course this is only valid under the same conditions as 
Berezin's global result. 
 
\section {Lie supergroup structure}
Here we explain Berezin's construction of an analytic Lie supergroup
associated to a Lie superalgebra $\mathfrak g=\mathfrak g_0+\mathfrak g_1$.
For this the structure sheaf $\mathcal F$ is the sheaf of germs of holomorphic
maps of $G_0$ to the Grassmann algebra $\wedge \mathfrak g_1^*$ which
we write as $\mathcal O_{G_0}\otimes \wedge \mathfrak g_1^*$. 
Since we are primarily interested in invariant differential operators,
we concentrate on construction of the representations of $\mathfrak g$ 
as derivations on $\mathcal F$ which correspond to left and right 
multiplication in the classical Lie group case. 
\subsection {Grassmann envelope}
The first step for the construction is to consider the graded
tensor product $\mathfrak g\otimes \Lambda $ with an arbitrary
(finite-dimensional) Grassmann algebra. We assume throughout
that $\Lambda $ is isomorphic to the Grassmann algebra of 
an $N$-dimensional vector space with $N\ge n:=\mathrm {dim}(\mathfrak g_1)$
and speak of $\Lambda $ as being generated by $N$ (odd) independent 
elements.
$$
(\mathfrak g\otimes \Lambda )_0=
\mathfrak g_0\otimes \Lambda _0\oplus \mathfrak g_1\otimes \Lambda _1
$$
and
$$
(\mathfrak g\otimes \Lambda)_1=
\mathfrak g_0\otimes \Lambda _1\oplus \mathfrak g_1\otimes \Lambda _0\,.
$$
For homogeneous elements $X,Y\in \mathfrak g$ and 
$\alpha ,\beta \in \Lambda $ one defines 
$$
[\alpha X,\beta Y]:=
(-1)^{\vert X\vert \vert \beta \vert}\alpha\beta [X,Y]\,.
$$
Extending by linearity, this defines a Lie superalgebra structure
on the tensor product.  Equipped with this structure, the
\emph{Grassmann envelope}
$\mathfrak g (\Lambda ):=(\mathfrak g\otimes \Lambda )_0$ is
a usual Lie algebra. Note that in addition to being
a Lie algebra, $\mathfrak g(\Lambda)$ is a $\Lambda _0$-module.
For notational convenience we write the multiplication by
elements of $\Lambda $ on the left, i.e., $X\otimes \alpha =:
(-1)^{\vert X\vert \vert \alpha\vert}\alpha X$.
Although we regard $\Lambda $ as a variable, we suppress it notationally
by letting $\tilde {\mathfrak g}:=\mathfrak g(\Lambda )$.
\subsubsection* {Decomposition of $\mathbf{\tilde {\mathfrak g}}$}
The Lie algebra $\mathfrak g_0$ is a subalgebra of 
$\tilde {\mathfrak g}$ which has a complementary ideal $\mathfrak k$
which is generated by the homogeneous elements $\alpha X$ where
$\mathrm {deg}(\alpha )\ge 1$.  We write 
$\tilde {\mathfrak g}=\mathfrak g_0\ltimes \mathfrak k$ as a semidirect
sum. Let $\tilde {\mathfrak g}_0$ denote the Lie subalgebra of 
$\tilde {\mathfrak g}$ which is generated by the homogeneous elements
$\alpha X$ where $\vert \alpha \vert =0$, i.e., where $\alpha $ is
even.  It follows that 
$\tilde {\mathfrak g}_0=\mathfrak g_0\ltimes \mathfrak k_0$, where
$\mathfrak k_0:=\mathfrak {k}\cap (\mathfrak {g}_0\otimes \Lambda _0)$
is the subalgebra of consisting even elements in the nilpotent
Lie algebra $\mathfrak k$.  Observe that the linear subspace
$\mathfrak k_1$ of odd elements of $\mathfrak k$ is stabilized
by $\tilde {\mathfrak g}_0$, i.e., 
$[\tilde {\mathfrak g}_0,\mathfrak k_1]\subset \mathfrak k_1$. 
\subsubsection* {Decomposition at the group level}
Let $K$ be the simply connected Lie group associated to the
Lie algebra $\mathfrak k$.  Recall that 
$\mathrm {exp}:\mathfrak k\to K$ is a diffeomorphism. If we
embed $\mathfrak k$ as a Lie algebra of matrices, $\mathrm{exp}$ is even
polynomially defined with a polynomial inverse.  In particular
it is bianalytic or biholomorphic, depending on the setting at
hand. Let $G_0$ be a (connected) Lie group associated to 
$\mathfrak g_0$ which integrates the representation
of $\mathfrak g_0$ on $\mathfrak g_1$. It follows that
$G_0$ acts on $K$ and $\tilde G:=G_0\ltimes K$ is associated
to $\tilde {\mathfrak g}$. Since $G_0$ acts on the Lie group 
$K_0:=\mathrm {exp}(\mathfrak k_0)$, the semidirect product
$\tilde G_0:=G_0\ltimes K_0$ is a closed Lie subgroup of 
$\tilde G$.  The linear space $\mathfrak k_1$ of odd elements
of $\mathfrak k$ defines a submanifold 
$K_1:=\mathrm {exp}(\mathfrak k_1)$ of $K$ which is stabilized
by $\tilde G_0$-conjugation. We write $\tilde G=\tilde {G}_0K_1$ which
is in a certain sense also a semidirect product.

\subsection {Grassmann canonical coordinates}
Here, starting with canonical coordinates on $G_0$, we construct
Grassmann canonical coordinates on $\tilde G$. 
\subsubsection* {Canonical coordinates on $\mathbf{G_0}$}

Let $U$ be an open
neighborhood of $0\in \mathfrak g_0$ so that 
$\mathrm {exp}:U\to G_0$ is a diffeomorphism onto an open neighborhood
$V$ of $\mathrm {Id}\in G_0$. For $g\in G_0$ we have the open
neighborhood $V.g$ which is identified with the coordinate chart
$U$ via the (bianalytic/biholomorphic) diffeomorphism 
$\mathrm {exp}:U\to V$. We shrink $U$ to $U_1$ so that products
of two elements as well as inverses of elements in $V_1=\mathrm {exp}(U_1)$ 
are still contained in $V$ and can still be identified with elements of $U$.
To simplify notation, the possibly smaller sets $U_1$ and $V_1$
are still denoted by
$U$ and $V$, respectively.  
We cover $G_0$ by neighborhoods of this
form which satisfy the further condition that if $Vg_1$ and $Vg_2$
have nonempty intersection, then $g_1g_2^{-1}\in V$.

\bigskip\noindent
In order to obtain concrete coordinates we let $\{X_1,\ldots X_m\}$
be a basis of $\mathfrak g_0$ and to $X=\sum x_iX_i\in U$ associate
the $m$-tuple $x=(x_1,\ldots ,x_m)$.  If $g\in Vg_1\cap Vg_2$ has
coordinate $x$ with respect to the chart $Vg_1$ and $z$ with
respect to $Vg_2$, and $g_1g_2^{-1}=\mathrm {exp}(X)$, then
change of coordinates is computed by applying the 
Campbell-Baker-Hausdorff formula (CBH) to
$$
\mathrm {exp}(\sum x_iX_i)\mathrm {exp}(X)=\mathrm {exp}(\sum z_iX_i)\,.
$$
\subsubsection* {Grassmann coordinates on $\mathbf{\tilde {G}_0}$}
Recall that $\mathrm {exp}:\mathfrak k_0\to K_0$ is a diffeomorphism.
Thus for $U$ as above we can view $U\times \mathfrak k_0$ as
a coordinate neighborhood of $\mathrm {Id}$: 
Let $x=(x_1,\ldots ,x_m)$ be as above and $h=(h_1,\ldots ,h_m)$
be an $m$-tuple of elements $h_i\in \Lambda $.  Then coordinates
on $V\times K_0$ are given by the diffeomorphism
$U\times \mathfrak k_0\to VK_0\cong V\times K_0$,
$(x,h)\to \mathrm {exp}(\sum x_iX_i)\mathrm {exp}(\sum h_iX_i)$.
Using the following remark, we will express these coordinates
in a more convenient way.
\begin {proposition}
The restriction of the $\tilde G_0$-exponential map to 
$U\times \mathfrak k_0$ is a diffeomorphism onto its image
$VK_0$. 
\end {proposition}
\begin {proof}
For $(A,B)\in U\times \mathfrak k_0$ with $B$ sufficiently small, i.e.,
sufficiently near $0\in \mathfrak {k}_0$,
we apply CBH to obtain a mapping 
$B\to \beta (A,B)$ defined
by $\mathrm {exp}(A+B)=\mathrm {exp}(A)\mathrm {exp}(\beta (A,B))$
Now the definition of $\beta $ involves brackets where 
$B$ appears in a given (higher order) bracket at most $N$ times,
where $\Lambda $ is the Grassmann algebra of a vector space
of dimension $N$.  If we use the coordinate $h$ to describe $B$,
this means that for fixed $A$ the mapping $\beta $ is polynomial
in $h$. So we fix $A$ and regard $\beta $ as a map 
$\beta :\mathfrak k_0\to \mathfrak k_0$.  Analogously we define
$\alpha :\mathfrak k_0\to \mathfrak k_0$ by 
$\mathrm {exp}(A)\mathrm {exp}(C)=\mathrm {exp}(A+\alpha (C))$
in the range where CBH applies and then extend by the polynomial
property.  In the range where CBH applies we have
$\mathrm {exp}(A+B)=\mathrm {exp}(A+\alpha\beta (B))$ and therefore
$\alpha\beta (B)=B$ for $B$ sufficiently small. Thus it follows
from the polynomial property that $\alpha\beta =\mathrm {Id}$.
Arguing in the same way,
$$
\mathrm {exp}(A)\mathrm {exp}(C)=
\mathrm {exp}(A+\alpha (C))=\mathrm {exp}(A)\mathrm {exp}(\beta \alpha(C))
$$
implies that $\beta\alpha (C)=C$ for $C$ sufficiently small and it
follows from the polynomial property that $\beta\alpha =\mathrm {Id}$.
Hence $(A,B)\to \mathrm {exp}(A+B)$ is a bijective map 
$U\times \mathfrak k_0\to VK_0$ and the desired result follows
from its analyticity.  In fact $\beta =\beta (A,B)$ is a polynomial map 
in $B$ with coefficients analytic in $A$.
\end {proof}
Now we cover $G_0$ with neighborhoods $Vg$ as before and
as a result have coordinates $(x,h)$ on $VK_0g=VgK_0$ given by
$(x,h)\to \mathrm {exp}(\sum (x_i+h_i)X_i)g$. Note that in all considerations
the Grassmann variables $z_i=x_i+h_i$ behave as scalars.
\subsubsection* {Canonical coordinates on $\mathbf{\tilde G}$}
Now let $\{\Xi_1,\ldots ,\Xi_n\}$ be a basis 
of $\mathfrak g_1$. An $n$-tuple $\xi =(\xi_1,\ldots ,\xi_n)$
with $\xi_j\in \Lambda _1$ defines $\sum \xi_j\Xi_j\in \mathfrak k_1$ and,
allowing $\xi $ to be arbitrary, we view it as a coordinate on 
$\mathfrak k_1$.  The exact same argument as that above
shows that $(x,h,\xi)\to \mathrm {exp}(\sum (x_i+h_i)X_i+\sum \xi_j\Xi_j)$
defines a diffeomorphism (analytic/biholomorphic) 
$U\times \mathfrak k\to VK$.  Again covering $G_0$ as above, we
cover $\tilde G$ with neighborhoods $VKg$ on which we have
the coordinates $(x,h,\xi)\to 
\mathrm {exp}(\sum (x_i+h_i)X_i+\sum \xi_j\Xi_j)g$.  These are the
promised \emph{Grassmann canonical coordinates}. 
\subsection {The $\mathbf{\mathfrak g}$-representation on
$\mathbf{\mathcal O_{\tilde G}\otimes \Lambda}$.}
For $X$ a homogeneous element of $\mathfrak g$, let $\alpha $
be a homogeneous element of $\Lambda $ with $\alpha=1$ if 
$\vert X\vert =0$ and $\vert \alpha \vert=1$ if $\vert X\vert =1$
so that in particular $\alpha X\vert=0$,  and
consider the 1-parameter subgroup $t\mapsto \mathrm {exp}(t\alpha X)$
in the Lie group $\tilde G$.  For $t$ sufficiently small we use
CBH to express its action by left multiplication in Grassmann 
canonical coordinates:
\small{
$$
\mathrm {exp}(t\alpha X)\mathrm {exp}(\sum (x_i+h_i)X_i+\sum \xi_j\Xi_j)g
=\mathrm {exp}(\sum (x_i+h_i)X_i+\sum \xi_j\Xi_j + t\alpha \eta)g\,.
$$} 
\noindent
For this recall that the CBH computation that produces the righthand
side of this equation involves higher order brackets where in theory
$\alpha X$ can appear any number of times.  However, using the
definition of the Lie bracket, in particular that the Grassmann
variables are behaving as scalars, along with the fact that $\alpha ^2=0$,
it follows that the only nonzero terms are those where $\alpha X$
appears exactly once.  

\bigskip\noindent
Now, for homogeneous
elements $Y,Z\in \mathfrak g_0$ and $\gamma ,\delta \in \Lambda $,
the bracket $[\gamma Y,\delta Z]$ is defined as 
$(-1)^{\vert Y\vert \vert \delta \vert}\gamma \delta [Y,Z]$. But
since $\tilde {\mathfrak g}$ is the even part of 
$\mathfrak g\otimes \Lambda $, it follows that 
$\vert \delta \vert =\vert Z\vert $. Thus, when transporting
$t\alpha $ to the front of a bracket of homogeneous elements
which occurs above, one only pays the price of a sign which does
not depend on $\alpha $. Therefore, for a fixed coordinate chart
$VK.g$, the element $\eta \in {\mathfrak g}\otimes \Lambda $ is well-defined
independent of $\alpha $. 

\bigskip\noindent
Now let $\mathcal L_{\alpha X}$ be the Lie derivative defined by left
multiplication by the 1-parameter group $\mathrm {exp}(t\alpha X)$.
\begin {proposition}
There is an operator $M_X$ on $\mathcal O_{G_0}\otimes \Lambda $ so
that $\mathcal L_{\alpha X}=\alpha M_X$.
\end {proposition}
\begin {proof}
Let the covering be denoted by $\{VKg_\ell\}$. As we just observed
there exist operators $M_X^\ell$ on functions
on $VKg_\ell$ which are defined independent of $\alpha $
so that $\mathcal L_{\alpha X}^\ell =\alpha M^\ell_X$. Here
$\mathcal L_{\alpha X}^\ell $ denotes the expression of the
globally defined Lie derivative in the local coordinates at
hand. Regarding this as a vector field it obeys the change
of coordinates rule $\mathcal L_{\alpha X}^k=J_{k\ell}\mathcal L_{\alpha X}^\ell$,
where $J_{k\ell}$ is the Jacobian of the change of coordinates.
Therefore $\alpha (M_X^k-J_{k\ell}M_X^\ell)=0$ for every $\alpha \in \Lambda $.
Now we embed $\Lambda $ in a larger Grassmann algebra $\Lambda '$ so
that if $\beta \in \Lambda $ is such that $\alpha \beta =0$ for all
$\alpha $ in $\Lambda '$, then it follows that $\beta=0$. We carry
out the above construction for $\Lambda '$ but, with the exception
of $\alpha $, all computations are made with elements of
$\tilde {\mathfrak g}=\mathfrak g(\Lambda)$.  It follows that
applying the difference  $M_X^k-J_{k\ell}M_X^\ell$ to an element
of $\mathcal O_{\tilde G}\otimes \Lambda $ results in a function
with values $\beta \in \Lambda $ with $\alpha\beta =0$ for
all $\alpha \in \Lambda '$.  Consequently, as desired
$M_X^k=J_{k\ell}M_X^\ell$ as operators on 
$\mathcal O_{\tilde G}\otimes \Lambda $.    
\end {proof}
\subsubsection* {Representation theoretical properties of $\mathbf{M}$}
Define $M:\mathfrak g\to \mathrm {End}(\mathcal O_{G_0}\otimes \Lambda)$
by $X\mapsto M_X$.  Although the notation may not indicate it,
$M_X$ is an operator at the sheaf level.  Observe that if $X$ is
homogeneous, then so is $M_X$ and $\vert M_X\vert =\vert X\vert $.  
The main properties of $M$ can be summarized as follows.
\begin {proposition}
The mapping $M$ is a representation of the Lie superalgebra 
$\mathfrak g$ in the space 
$\mathrm {Der}(\mathcal O_{\tilde G}\otimes \Lambda)$ of derivations of 
$\mathcal O_{\tilde G}\otimes \Lambda $.
\end {proposition}
\begin {proof}
In the discussion 
$\alpha$ and $\beta $ are chosen to be homogeneous elements
in some larger Grassmann algebra $\Lambda '$
with the property that $\alpha\beta\gamma =0$ implies that $\gamma =0$
for every $\gamma \in \Lambda $. For $X,Y$ homogeneous elements
in $\tilde {\mathfrak g}=\mathfrak g(\Lambda)$, we 
carry out the above construction for 
$\alpha X,\beta Y\in \mathfrak g(\Lambda ')$.  Since
$\vert \alpha\vert =\vert X\vert$ and $\vert \beta \vert =\vert Y\vert $,
it follows immediately from the definition of the Lie bracket that
$$
\mathcal L_{[\alpha X,\beta Y]}=
(-1)^{\vert M_X\vert \vert M_Y\vert}\alpha\beta M_{[X,Y]}\,.
$$
On the other hand,
\begin {gather*}
\mathcal L_{[\alpha X,\beta Y]}=
\mathcal L_{\alpha X}\mathcal L_{\beta Y}-
\mathcal L_{\beta X}\mathcal L_{\alpha Y}=\\
\beta \alpha M_XM_Y-\alpha\beta M_YM_X=
\alpha\beta ((-1)^{\vert M_X\vert \vert M_Y\vert}M_XM_Y-M_YM_X)
\end {gather*}
and therefore
$$
\alpha \beta (M_{[X,Y]}-(M_XM_Y-(-1)^{\vert M_X\vert \vert M_Y\vert}M_YM_X)=0\,.
$$
Applying this identity to 
$f\in \mathcal O_{\tilde G}\otimes \Lambda $ and using the cancellation
property of $\alpha\beta $ yields the representation property of
$M$.

\bigskip\noindent
In order to show that $M_X$ is a derivation, we choose $\alpha \in \Lambda '$
with the cancellation property and note that for 
$f,g\in \mathcal O_{\tilde G}\otimes \Lambda $
\begin {gather*}
\alpha M_X(fg)=\mathcal L_{\alpha X}(fg)=\\
\mathcal L_{\alpha X}(f)g+f\mathcal L_{\beta X}g=
\alpha M_X(f)g+f\alpha M_X(g)=\\
\alpha (M_X(f)g+(-1)^{\vert f\vert \vert X\vert}fM_X(g))
\end {gather*}
and the desired derivation property follows by cancellation.
\end {proof}
Before closing this section, we should note that the
same discussion as above with left multiplication 
by the 1-parameter groups $t\mapsto \mathrm {exp}(t\alpha X)$ 
being replaced by right multiplication leads
to analogous representations of ${\mathfrak g }$ in
$\mathrm {Der}(\mathcal O_{\tilde G}\otimes \Lambda )$.  If it
is necessary to differentiate between the representations defined
by left multiplication and those defined by right multiplication,
we denote the former by $X\mapsto L_X$ and the latter by $X\mapsto R_X$.
Of course these representations commute in the graded sense,
i.e.,$L_XR_Y=(-1)^{\vert X\vert \vert Y\vert}R_YL_X$ and the representation
which corresponds to conjugation in the classical Lie group case is given
by $X\to L_X+R_X$.
\subsection {The $\mathbf{\mathfrak g}$-representation on 
$\mathbf{\mathcal O_{G_0}\otimes \wedge {\mathfrak g}_1^*}$}
Above we have constructed $\mathfrak g$-representations on
the function algebra of the Grassman envelope $\tilde G$.  Here
we turn to the main task of this section which is to construct
the representations of $\mathfrak g$ on the structure sheaf 
$\mathcal F=\mathcal O_{G_0}\otimes \Lambda$ which correspond 
to the representations by invariant vector fields of a 
Lie algebra on the structure sheaf of an associated Lie group.
The key idea is to extend (analytic/holomorphic) functions
from $G_0$ to functions on the special Grassmann envelope $\tilde G$
where $\Lambda =\wedge {\mathfrak g}_1^*$ and apply the representations
constructed above.

\subsubsection* {Grassmann analytic continuation}
Let $\Lambda =\wedge {\mathfrak g}_1^*$ and consider a 
$\Lambda $-valued analytic or holomorphic function 
$f$ on an open subset $V$ of $G_0$. We (analytically) continue $f$
to a function $\Psi (f)$ on the open subset $VK$ of the Grassmann
envelope $\tilde G$.  Since the construction does not depend
on the nature of $V$, we suppress it in the discussion and only
discuss functions defined on the full group $G_0$. 

\bigskip\noindent
As a first step we will (analytically) continue $f$ 
to a function on $\tilde G_0$.  
For this we cover $G_0$ as usual by open sets of the form 
$Vg$ where $V$ is an open neighborhood of $\mathrm {Id}$ 
with $\mathrm {exp}:U\to V$ a diffeomorphism. Letting
$\{X_1,\ldots,X_m\}$ be a basis of $\mathfrak g_0$ we have
Grassmann coordinates on $VK_0g$ given by 
$$
(x,h)\mapsto \mathrm {exp}(\sum (x_i+h_i)X_i)
$$   
on $VK_0$ and then composing by multiplication on $VK_0g$. Let
$g_k$ and $g_\ell$ be such that the coordinate neighborhoods
$Vg_k$ and $Vg_\ell$ have nonempty intersection.  Since by assumption 
$f$ is analytic (or holomorphic), in the respective coordinate charts 
it has convergent power series representation $f_k(x)$ and 
$f_\ell(x)$.  Now a given point $v\in Vg_k\cap Vg_\ell$ is
represented as
$$
v=\mathrm{exp}(\sum x_i^kX_i)g_k=\mathrm{exp}(\sum x_i^\ell X_i)g_\ell\,.
$$
with $g_\ell g_k^{-1}=\mathrm {exp}(X)\in V$.  As we underlined above, 
if $(x,h)$ are Grassmann coordinates defined by 
$$
(x,h)\mapsto \exp(\sum (x_i+h_i)X_i)g
$$
where $g=g_k$ or $g=g_\ell$, and the change of variables at
the level of $G_0$, which is computed by CBH, 
is given by $x^k=A_{k\ell}(x^\ell)$, then the change of variables 
for sum $x+h$ is given by the same rule, i.e..
$x^k+h^k=A_{k\ell}(x^\ell+h^\ell)$. 

\bigskip\noindent
Now, since the elements $h$ of the Grassmann algebra are nilpotent, 
the power series $f_k(x+h)$ and $f_\ell(x+h)$ converge 
on $U\times \mathfrak k_0$, and since the transformation rule
for the variable $x$ is the same as that for $x+h$, the resulting
locally defined functions agree on the intersection $VK_0g_k\cap VK_0g_\ell$.
Thus we have defined the basic first step of
\emph{Grassmann analytic continuation}
which is a continuous morphism of sheaves of algebras
$$
GAC: \mathcal O_{G_0}\otimes \Lambda 
\to \mathcal O_{\tilde G_0}\otimes \Lambda\,.
$$
The next step, i.e., the continuation to 
$\mathcal O_{\tilde G}\otimes \Lambda $,
is formal.  For this we observe that $\tilde G$ is 
the product $\tilde G_0\times \mathfrak k_1$.  Using
a basis $\{\Xi_1,\ldots ,\Xi_m\}$ for $\mathfrak g_1$,
an element of $\mathfrak g_1\otimes \Lambda _1$ is written
as $\sum \sigma _j\Xi_j$.  Thus the product structure 
is given by 
$(\tilde g_0,\sigma)\mapsto \tilde g_0\mathrm {exp}(\sum \sigma _j\Xi_j)$.
But the dependence of $\mathrm {exp}(\sum \sigma _j\Xi_j)$ on the $\sigma _j$
is polynomial. Consequently, if $\{\Xi_1^*,\ldots ,\Xi_m^*\}$
is the dual basis of $\mathfrak g_1^*$, then an arbitrary function
$f\in \mathcal O_{\tilde G}\otimes \Lambda $ can be expressed
as 
\begin {equation*}
f=\sum_{\vert I\vert \le m} f_I\Xi^*_I=
\sum_{i_1<\ldots <i_p}f_{i_1,\ldots ,i_p}
\Xi_{i_1}^*\wedge \ldots \wedge \Xi^*_{i_p}\,,
\end {equation*}
where the coefficient functions are arbitrary analytic/holomorphic
functions on $\tilde G_0$.

\bigskip\noindent
Using the above description, the continuation of a function
from $\tilde G_0$ is formal. Recalling that 
$\Lambda :=\wedge {\mathfrak g}_1^*$, a function
$f\in \mathcal O_{\tilde G_0}\otimes \Lambda $ is already of the
form 
$$
f=\sum_{\vert I\vert \le m} f_I\Xi_I^*\,.
$$
Its continuation $\widehat f\in \mathcal O_{\tilde G}\otimes \Lambda $
is simply defined as
$$
\widehat f:=\sum_{\vert I\vert \le m} f_I\Xi^*_I\,,
$$
re-interpreting the elements $\Xi_j\in \mathfrak {g}_1^*$ of the dual
basis as odd coordinate functions on $\tilde G$.
The Grassmann analytic continuation from $\mathcal O_{G_0}\otimes \Lambda $
is then defined by
$$
\Psi :\mathcal O_{G_0}\otimes \Lambda \to 
\mathcal O_{\tilde G}\otimes \Lambda , \ f\mapsto \widehat {GAC(f)}\,.
$$ 
The following summarizes the construction. Recall that for this
$\Lambda :=\wedge {\mathfrak g}_1^*$.
\begin {proposition}
Grassmann analytic continuation $\Psi :\mathcal O_{G_0}\otimes \Lambda \to 
\mathcal O_{\tilde G}\otimes \Lambda$ is an algebra morphism which is
an isomorphism onto its image $\mathcal A$. The algebra $\mathcal A$
is the set of functions $F\in \mathcal O_{\tilde G}\otimes \Lambda $ 
of the
form 
$$
F=\sum_{\vert I\vert \le m}f_I\Xi^*_I
$$
where $f_I=\mathrm {GAC}(h_I)$ for $h_I\in \mathcal O_{G_0}\otimes \Lambda $.
\end {proposition}
It should be remarked that a function $F=\sum f_I\Xi_I^*$ is in the image 
$\mathcal A$ if and only if the expansions of the coefficents
$f_I$ in the local coordinates $(x,h)$ are power series in $x+h$.
In other words, from the point of view of the coefficent functions
the Grassmann variables $x+h$ are scalars.

\subsubsection* {The representations of $\mathbf{\mathfrak g}$ on 
the structure sheaf}
Recall that at the beginning we deemed the structure sheaf 
$\mathcal F$ to be the sheaf of germs of holomorphic maps
from $G_0$ to $\wedge {\mathfrak g}_1^*$ which we denote here
by $\Lambda $.  By definition the numerical functions are 
just the analytic/holomorphic
functions $\mathcal O_{G_0}$ with the projection
$\mathrm {num} :\mathcal F\to \mathcal O_{G_0}$ being the obvious one.
Since $\mathcal O_{G_0}$ is canonically embedded in $\mathcal F$
with $\mathrm {ker}(\mathrm {num})$ as complement, in the language
of supergroups the structure is \emph{split}.

\bigskip\noindent
We complete the task of constructing the representations
of $\mathfrak g$ on $\mathcal F$ by noting the following.
\begin {proposition}
The image $\mathcal A$ of Grassmann analytic continuation is
invariant under the representations $L$ and $R$ of $\mathfrak g$
on $\mathcal O_{\tilde G}\otimes \Lambda $.
\end {proposition}
\begin {proof}
We must show that if $X$ is a homogeneous element of 
$\mathfrak g$ and 
\begin {equation*}\label{image}
F(\tilde g_0\mathrm {exp}(\sum_j\sigma_j \Xi_j))=\sum f_I\sigma_I
\end {equation*}
is in $\mathcal A$, then $L_X(F)$ and $R_X(F)$ are likewise
in $\mathcal A$. Therefore
we compute $\mathcal L_{\alpha X}(F)$ in Grassmann analytic 
coordinates. For example, in a neighborhood of $\mathrm {Id}$
we must use CBH to convert the product
$$
\mathrm {exp}(t\alpha X)\mathrm {exp}(\sum z_iX_i)
\mathrm {exp}(\sum \sigma _j\Xi_j)
$$
to the form of the Grassmann coordinates. Here $z=x+h$ is
the Grassmann variable. The result
is complicated, but is of the form
$$
\mathrm {exp}(\sum_i (t\alpha \tilde z_i(z,\sigma )+z_i)X_i)
\mathrm {exp}(\sum_j (t\alpha \tilde \sigma _j(z,\sigma)+\sigma _j)\Xi_i)\,.
$$
The key point is that the coefficients 
$\tilde z_i(z,\sigma )=\sum_Jc_i^J(z)\sigma _J$ are
polynomials in $\sigma $ with coefficients which are power series
in $z$.  Now recall that applying $F$ to such an expression
yields
$$
\sum f_I({exp}(\sum_i (t\alpha \tilde z_i(z,\sigma )+z_i)X_i))
((t\alpha \tilde \sigma (z,\sigma)+\sigma )_I\,.  
$$
Using the power series representations of the $f_I$ when
differentiating one obtains polynomials in $\sigma $ with
coeffiecients which are power series in the Grassmann
variables $z$. Of course factoring out $\alpha $ doesn't
change this structure and it follows that the resulting
function is of the desired form, i.e., 
$$
\mathcal L_{\alpha X}(F)=\alpha .\sum \tilde f_I\sum \Xi^*_I
$$
where the coefficients $f_I$ are power series in Grassmann
analytic coordinates.  Note that although we have only verified
this in the coordinates at the identity, the verification
in a general neighborhood $VK_1g$ only differs from that at
the identity by conjugation of $K_1$ by $g\in G_0$. 
\end {proof}
As a consequence we have the main result of this section.
\begin {theorem}
Restricting the representations $L$ and $R$ of $\mathfrak g$
on $\mathcal O_{\tilde G}\otimes \wedge {\mathfrak g}_1^*$ to
the image $\mathcal A$ of Grassmann analytic continuation 
$\Psi :\mathcal F\overset{\cong}{\to} \mathcal A$ defines
the representations of $\mathfrak g$ on the structure
sheaf $\mathcal F$. Restricted to $\mathfrak g_0$
these are given by differentiating elements of $\mathcal F$  
by the natural action of 1-parameter subgroups of $G_0$
on $\mathcal F$.  
\end {theorem}
\subsection {Comments on Lie supergroups}
A Lie supergroup is more than a structure sheaf $\mathcal F$ 
with representations $L$ and $R$ which extend the natural
representations of $\mathfrak g_0$.  The appropriate additional
structure can be formulated in terms of
a triple $(\mu, \iota ,\varepsilon)$
of maps at the sheaf level which correspond to
multiplication, inverse and the identity. For example,
in the classical case the multiplication map
$\mu :\mathcal O_{G}\to \mathcal O_{G\times G}$ is given
by $\mu(f)(g_1,g_2)=f(g_1g_2)$. Certain compatibility
conditions must be fulfilled (see \cite{B}, p.247). Once
this triple is defined, the representations $L$ and $R$
can be computed.  Using Grassmann analytic continuation
as above, Berezin does indeed define such a triple and then
constructs $L$ and $R$ as above. 

\bigskip\noindent
Berezin's construction of the Lie supergroup triple 
uses a notion of a superfunction which 
is slightly different from that above.  In simple terms
this is the difference between the notion of a variable
and the evaluation of the variable. For example, a complex 
polynomial $P$ in several commuting variables $z_1,\ldots ,z_m$ 
is an element of the ring $\mathbb C[z_1,\ldots, z_n]$.  If $R$ is any
commutative $\mathbb C$-algebra, e.g., of numbers, functions, operators, etc.,
$P$ can be evaluated to define a function 
$P_R:R\times \cdots \times R\to R$.  In the setting at hand
$\xi_1,\ldots, \xi_n$ are anticommuting variables and Berezin defines a
superfunction on $G_0$ to be an element 
$f\in \mathcal O_{G_0}[\xi_1,\ldots ,\xi_m]$.  This is a function 
$f=f(g,\xi)$ in two blocks of variables. The first variables $g$ can
be regarded as the commuting variables defined by the coordinates in
the manifold $G_0$, e.g., in some power series ring. The second variables
are the anticommuting variables which generate a Grassmann algebra
$\Lambda (\xi_1,\ldots ,\xi_m)$.

\bigskip\noindent
Continuing with the simple example, if $a_1,\ldots a_n$ are algebraically
independent elements of a commutative ring $R$, then any algebraic
combination $\sum c_Ia^I$ defines a unique polynomial $P\in R[z_1,\ldots, z_n]$.
In our context where $n:=\mathrm {dim}(\mathfrak g_1)$, one considers
$\mathcal O_{G_0}\otimes \Lambda $ where $\Lambda $ 
is a Grassmann algebra generated by $N$ elements with
$N\ge m$. If $a_1,\ldots, a_m$ are independent elements of $\Lambda $,
then there is a canonical isomorphism between the subalgebra spanned
by these elements and the algebra of superfunctions 
$\mathcal O_{G_0}[\xi_1,\ldots ,\xi_m]$. 

\bigskip\noindent
As stated at the outset we wish to regard the structure sheaf
$\mathcal F$ as the sheaf of germs of holomorphic functions on
$G_0$ with values in $\wedge \mathfrak {g}^*_1$.  This may seem
different from Berezin's sheaf, but if one exchanges variables for
evaluated variables, it is not.  The main point in this regard is
that Berezin's construction uses a basis $\{\Xi_1,\ldots ,\Xi_n\}$
for $\mathfrak {g}_1$. Thus evaluating a Berezin superfuction
$f=f(g,\xi)\in \mathcal O_{G_0}[\xi_1,\ldots ,\xi_n]$ on 
$\Xi^*_1,\ldots ,\Xi^*_n$ defines such a function. Conversely, 
starting with a holomorphic function $f:G_0\to \wedge \mathfrak {g}^*_1$
and the basis $\{\Xi_1,\ldots ,\Xi_n\}$ one obtains a Berizin
superfunction by replacing the $\Xi_j^*$ by the variables $\xi_j$.   

\bigskip\noindent
When defining the Lie supergroup triple $(\mu, \iota ,\varepsilon)$
the variable viewpoint is very useful. For example, the standard
multiplication morphism $\mu :G_0\to G_0\times G_0$ must be
lifted to a map $\varphi :\mathcal F_{G_0}\to \mathcal F_{G_0\to G_0}$
of sheaves.  This map can be defined as follows.  We regard
a function in $\mathcal F_{G_0}$ as a function of one variable
$(g,\xi)$.  Here $g$ is already Grassmann continued to $\tilde G_0$
where $\Lambda $ is generated by $N$ odd elements with $N\ge 2n$.
As usual $\xi $ denotes an $n$-tuple of
independent Grassmann variables.  The functions in
$\mathcal F_{G_0\times G_0}$ are functions of two such variables
$((g_1,\xi_1),(g_2,\xi_2))$. Now we \emph{evaluate} these
variables at points in the group $\tilde G$ by using
$2n$ independent odd elements $a_1,\ldots ,a_n,a_{n+1},\ldots ,a_{2n}$
in $\Lambda$.  The pair of variables defines a pair of group
elements  
$$
g_1\mathrm {exp}(\sum a_j\Xi_j),\
g_2\mathrm {exp}(\sum a_{n+j}\Xi_j)\in \tilde G\,.
$$
Using the decomposition $\tilde G=\tilde G_0K_1$, the product of
these group elements is written as $g\mathrm {exp}(\sum \eta _j\Xi_j)$.
Now we change back to the variable standpoint and regard $g$
as a Grassmann analytically continued variable in $\tilde G_0$
and $\eta =(\eta_1,\ldots ,\eta_n)$ as a variable which (as can be
checked) consists of $n$ independent odd Grassmann variables.  Since
the construction of manipulating the product to the form
$g\exp(\sum _j\eta_j\Xi_j)$ only uses the \emph{variable properties}
of the elments $a_j$, the \emph{product variable} $(g,\eta)$ does
not depend on the choice of the $a_j$. The multiplication morphism
at the sheaf level is then defined in the natural way:
$$
\varphi (f)((g_1,\xi_1),(g_2,\xi_2)):=f(g,\eta)\,.
$$
Lifting the inverse mapping $g\mapsto g^{-1}$ to the sheaf level
is defined analogously. The evaluation at the identity is
given by $\varepsilon (f)(g,\xi)=f(\mathrm{Id},0)$.  Since these
operations are defined by the group structure of $\tilde G$, it
can be directly checked that they have the desired compatibility
properties.  As indicated
at the outset of this paragraph, if we regard $\mathcal F$ as
the sheaf of $\wedge \mathfrak {g}_1^*$-valued holomorphic functions,
then the Lie supergroup structure is defined at that
level by replacing a function $\sum f_I(g)(\Xi^*)^I$ by
a Berezin function $f=f(g,\xi)$ of two variables defined
by $f(g,\xi)=\sum_If_I(g)\xi^I$.  Multiplication is then
defined as above at the level of variables and then one returns
to $\wedge \mathfrak {g}_1^*$-valued holomorphic functions
by evaluating the variables at $\Xi_1^*,\ldots ,\Xi_n^*$.

\bigskip\noindent
The representations $L$ and $R$ of $\mathfrak {g}$ on $\mathcal F$
are in fact by derivations which are invariant under the group
structure.  For example, invariance by right-multiplication means
that 
\begin {equation}\label{invariance}
(M_X\otimes \mathrm {Id})\circ \varphi =\varphi \circ M_X
\end {equation}
for $X\in \mathfrak {g}$.  This is immediate for the derivations
defined by even elements, because the supergroup multiplication
by elements of $\tilde G_0$ is just usual multiplication.  For
odd elements $X\in \mathfrak {g}_1$ one must check that 
the formula analogous to $\ref{invariance}$ with $M_X$ replaced
by the Lie derivative $\mathcal L_{\alpha X}$ in $\tilde G$ 
makes sense and is valid. Then one cancels $\alpha $ in the same way that
was done for the definition of $M_X$.   
     
\bigskip\noindent
There are various other methods of constructing a
Lie supergroup associated to $\mathfrak g$, (see \cite{K} and \cite{Ka2}
for constructions using Lie-Hopf algebras and \cite {Ko} for the
dual construction at the level of the structure sheaf).
But it turns out that, even in the more delicate holomorphic 
setting, these are all equivalent (see \cite {V}).
For further discussion of this matter, e.g., for a detailed comparison of
the various definitions, see \cite {Ka1}. 
\section {Radial operators}
Our concrete goal here is to describe
certain series of radial differential operators in 
sufficiently concrete form for computational
applications in \cite {CFZ} and \cite {HPZ}.  
As in the previous section, the main results were
originally proved by Berezin.
\subsection {Basic definitions}
Here we recall be the basic objects for the study of
radial differential operators on Lie supergroups.
\subsubsection* {Universal enveloping algebra}
The \emph{universal enveloping algebra} of a Lie superalgebra
$\mathfrak g$ is the quotient $U(\mathfrak g)$ of the full tensor algebra 
$T(\mathfrak g)$ by the ideal generated
by $(X\otimes Y-(-1)^{\vert X\vert \vert Y\vert }Y\otimes X)-[X,Y]$
where $X,Y$ are homogeneous elements of $\mathfrak g$ regarded
as elements of $(\mathfrak g)$. If
$X_1\otimes X_2\cdots \otimes X_k$ is a monomial in $T(\mathfrak g)$,
then its image in $U(\mathfrak g)$ is equipped with the sign
$(-1)^{\vert X_1\vert+\ldots +\vert X_k\vert}$. This defines a 
$\mathbb Z_2$-grading on $U(\mathfrak g)$ for which the induced
bracket defines a Lie superalgebra structure.  Here we shall 
primarily be concerned with the center $Z(\mathfrak g)$ of
$U(\mathfrak g)$. This is the subalgebra of $U(\mathfrak g)$ 
consisting of those elements $X$ with $\mathrm {ad}(X)=[X,\cdot ]=0$.
\subsubsection* {Laplace-Casimir operators}
Recall that a representation of a Lie superalgebra $\mathfrak g$
is a superalgebra morphism $\rho : \mathfrak {g}\to \mathrm {End}(V)$
to the superalgebra of linear maps of a graded vector space
$V=V_0\oplus V_1$. Such a representation extends to a representation
$\rho :U(\mathfrak g)\to \mathrm {End}(V)$. The examples of main
importance here are the representations $L$ and $R$ of 
$\mathfrak g$ constructed in the previous section on the 
structure sheaf $\mathcal F$ of superfunctions of the
associated Lie supergroup.  In particular, we consider the
representation $X\mapsto L_X$ and extend it to a representation
of $U(\mathfrak {g})$ by differential operators.  The 
\emph{Laplace-Casimir operators} are those in the image
of the center $Z(\mathfrak g)$.   
\subsection {Radial functions}
A superfunction $f$ on $G_0$ is said to be \emph{radial} if 
it is annihilated by the representation $X\to L_X+R_X$
of $\mathfrak g$.  Let $\mathcal R_{G_0}$ denote the sheaf of radial holomorphic
functions on $G_0$. Observe that if $D$ is a Laplace-Casimir
operator, then $D\vert \mathcal R_{G_0}:\mathcal R_{G_0}\to \mathcal R_{G_0}$.
We regard $\mathcal R_{G_0}$ as the natural domain of definition
of these operators.

\bigskip\noindent
Note that if $f$ is globally defined on $G_0$, then
the condition $L_Xf+R_Xf=0$ for all $X\in \mathfrak g_0$ just means
that $f$ is conjugation invariant, i.e., $f(g_0gg_0^{-1})=f(g)$
for all $g_0\in G_0$. If $G_0$ is reductive, which we assume
from now on, it follows that globally defined conjugation 
invariant superfunctions are completely determined by their 
restrictions to any given maximal torus $H$.  We fix such a 
maximal torus and let $\mathfrak h$ denote its Lie algebra in 
$\mathfrak g_0$. Let us also
assume, as will be the case in all applications,
that $\mathfrak h$ is a Cartan algebra of the full Lie superalgebra 
$\mathfrak g$.   
\subsubsection* {Restriction theorem at the Lie superalgebra level}
At the infinitesimal level we are interested in understanding
the image of the restriction map from the space 
$S(\mathfrak {g}^*)^{\mathfrak g}$ of $\mathrm {ad}_{\mathfrak g}$-invariant
(super) polynomials on $\mathfrak g$ to the space $S(\mathfrak {h}^*)$.
We view an element $P\in S(\mathfrak {g}^*)$ as a (holomorphic)
polynomial map $P:\mathfrak g_0\to \wedge \mathfrak {g}_1^*$.
If $U$ is a sufficiently small neighborhood of $0\in \mathfrak {g}_0$
which is identified by the exponential map with a neighborhood
$V$ of $\mathrm {Id}$ in $G_0$, then 
$P\in S(\mathfrak {g}^*)^{\mathfrak g}$ if and only if the resulting
function on $V$ is radial.  

\bigskip\noindent
Since it has been assumed that
$\mathfrak h$ is an even subspace of $\mathfrak g$, polynomials
in $S(\mathfrak {h}^*)$ are just standard (numerically valued)
polynomials. Since $P\in S(\mathfrak {g}^*)^{\mathfrak g}$ is invariant
by the adjoint representation of $G_0$, it is immediate that
$R(P)$ is invariant under the Weyl group 
$W=W(\mathfrak {g}_0,\mathfrak {h}_0)$.  Thus we regard $R$ as
a map $R:S(\mathfrak {g}^*)^{\mathfrak g}\to S(\mathfrak {h}^*)^W$.

\bigskip\noindent
A great deal is known about the restriction morphism $R$.  
In particular, it is always injective. In our
cases of interest, the basic results are proved in \cite {B}
and \cite {S}.  Let us quote Berezin's Theorem 3.1.
\begin {theorem}\label{extension}   
Let $G_0$ be reductive and assume that $\mathfrak h$ is a Cartan
algebra of $\mathfrak g$ which is contained in $\mathfrak g_0$.
Assume further that $\mathfrak g$ is endowed with a nondegenerate
invariant scalar product and that its odd root spaces are 
1-dimensional. Then a $W$-invariant polynomial is in the image
of $R$ if and only if for every odd root $\beta $ (with dual
root $\beta ^*$ ) it follows that
\begin {equation}\label{divisibility}
\frac{d}{dt}\Big\vert_{t=0}P(h+t\beta^*)=\beta (h)Q(h)
\end {equation}
where $Q=Q(h)$ is a polynomial on $\mathfrak h$.
\end {theorem}
We say that the \emph{extendible polynomials} are exactly those
which are $W$-invariant and satisfy the \emph{divisibility condition}
(\ref{divisibility}).  
Since our work here is aimed at understanding properties of
radial functions and operators in the cases of $\mathfrak {gl}$
and $\mathfrak {osp}$, it should be emphasized that the conditions
of Berezin's theorem are fulfilled in those cases.

\bigskip\noindent
For the statement of the version of Berizin's extension theorem 
for holomorphic functions we say that a holomorphic superfunction on
$\mathfrak g_0$ is radial if and only if it is annihilated
by the all $\mathrm {ad}_\mathfrak g$-derivations. The divisibility
condition for holomorphic functions is the same as that for
polynomials.  
\begin {theorem}
Under the assumptions of Theorem \ref{extension} 
it follows that a $W$-invariant function
$f\in \mathcal O(\mathfrak h)$ can be (uniquely) extended to
a radial holomorphic superfunction on $\mathfrak {g}_0$ if and only if
it satisfies the divisibility condition.
\end {theorem}
{\it Sketch of Proof}. For the \emph{neccessity}, i.e., that the
divisibility condition is really needed for extension, one
replaces the polynomial $P$ in Berezin's proof by the convergent
power series representation of the given function $f$ at 
$0\in \mathfrak h$.  For the \emph{sufficiency} Berezin uses
the fact that the given polynomial $P$ can be extended 
to a unique $\mathrm {Ad}_{G_0}$-invariant polynomial on $\mathfrak {g}_0$
and then proceeds by using generalities on superfunctions.  
Since $W$-invariant holomorphic functions on $\mathfrak h$ extend to 
$\mathrm{Ad}_{G_0}$-invariant holomorphic functions on $\mathfrak {g}_0$,
the same proof can be carried out in the holomorphic case.\qed
%
%
%\bigskip\noindent
%If $\tilde \Delta $ and $\tilde U$ as above and $\Delta $ and
%$U$ are open subsets which are defined in the same way by the
%$\mathrm {Ad}_{G_0}$-invariant theory, then of course the
%classical invariant theoretic extension of the restriction of
%a $W$-invariant holomorphic function $f$ on $\tilde \Delta $ is 
%the restriction of the extended function from $\tilde U$ to $U$.
%
%\bigskip\noindent
%{\bf Zusatz.} {\it In order for a $W$-invariant holomorphic function on
%$\Delta $ to be extended to a radial function on $U$ it is necessary 
%and sufficient for its restriction to $\tilde U$ to 
%satisfy the divisibility condition.}
%\begin {proof}
%One must only prove the sufficiency of the divisibility condition.  This
%follows by noting that the extended radial function on $\tilde U$ is
%defined in such a way, e.g., involving Grassmann analytic continuation,
%such that its domain of definition and holomorphy is the
%the domain $U$ to which the given function extends as an
%$\mathrm {Ad}_{G_0}$-invariant function. To prove that the extended function 
%is radial on $U$, one uses the fact that its retriction to
%$\tilde U$ is radial and applies the identity principle. 
%\end {proof}
% 
\subsubsection* {Restriction theorem at the group level}
Now let us turn to the Lie supergroup associated to $\mathfrak g$
equipped with its sheaf $\mathcal F$ of holomorphic superfunctions, i.e.,
the sheaf of germs of holomorphic maps with values in
$\wedge \mathfrak {g}_1^*$.  We assume that $\mathfrak g$ satisfies
the assumptions of Theorem \ref{extension} andlet $H$ be the maximal
(complex) torus in $G_0$ associated to $\mathfrak h$. In this situation
we wish to determine, e.g., the image of the restriction map
$R:\mathcal R_{G_0}(G_0)\to \mathcal O(H)^W$ from the globally
defined holomorphic radial functions on $G_0$ in the algebra
of $W$-invariant holomorphic functions on $H$.

\bigskip\noindent
For this the divisibility condition must be transferred to
the group level: A holomorphic function $f\in \mathcal O(H)$
is said to satisfy the divisibility condition if and only if
its pull-back $f\circ \mathrm {exp}$ satisfies the divisibility
condition on $\mathfrak h$.  The following is an immediate consequence
of the results in the previous paragraph.
\begin {proposition}  
A holomorphic function $f\in \mathcal O(H)^W$ satisfies the divisibility
condition if and only if its lift $f\circ \mathrm {exp}$ is the
restriction of a uniquely determined radial
holomorphic superfunction on $\mathfrak g$.
\end {proposition}
The extension theorem at the group level is stated as expected.
\begin {theorem}\label{group level extension}
Let $\mathfrak g$ be a Lie superalgebra which satisfies the
conditions of Theorem \ref{extension} and let $G_0$ be a base
of an associated Lie supergroup. Fix a Cartan algebra $\mathfrak h$
in $\mathfrak g$ and let $H=\mathrm {exp}(\mathfrak h)$.  Then
a $W$-invariant holomorphic function $f\in \mathcal O(H)$ is the restriction
of a radial holomorphic superfunction on $G_0$ if and only if
it satisfies the divisibility condition.
\end {theorem}
\begin {proof}
To prove the sufficiency of the divisibility condition we
let $f$ be a $W$-invariant (numerical) holomorphic function on $H$ which
satisfies the divisibility condition and $\hat f=f\circ \mathrm {exp}$.
Since $\hat f$ satisfies the divisibility condition on $\mathfrak h$,
it is the restriction of a (unique) holomorphic radial superfunction 
$E(\hat f)$ on $\mathfrak g_0$. This means that $\hat f$ is the numerical
part of $E(\hat f)\vert \mathfrak h$. The uniqueness of the extension
$E(\hat f)$ implies that $E(\hat f)\vert \mathfrak h$ is periodic
with respect to the discrete additive subgroup of $\mathfrak h$
which is the kernel of $\mathrm{exp}:\mathfrak h\to H$. Thus 
there is a $\wedge \mathfrak {g}_1 ^*$-valued
holomorphic function $f_s$ on $H$ with 
$E(\hat f)\vert \mathfrak h=f_s \circ \mathrm{exp}$.  

\bigskip\noindent
Observe that
since $E(\hat f)$ is invariant by conjugation by elements of
the normalizer of $H$, it follows that its restriction to $\mathfrak h$
is $W$-invariant and thus $f_s$ is $W$-invariant. Classical 
invariant theory then implies that $f_s$ is
the restriction of a unique conjugation invariant superfunction
$E(f)$ on $G_0$. The function $E(f)\circ {exp}$ is an 
$\mathrm {Ad}_{G_0}$-invariant holomorphic superfunction which
agrees with $E(\hat f)$ on $\mathfrak h$. Thus 
$E(f)\circ {exp}$ \emph{is} the radial extension of $\hat f$.
In particular, if $U$ is a neighborhood of $0\in \mathfrak {g}_0$
such that $\mathrm {exp}:U\to V$ is biholomorphic, then
the fact that the operators $L_X+R_X$ annihilate $E(\hat f)\vert U$
for all $X\in \mathfrak g$ implies that they annihilate $E(f)\vert V$
for all $X\in \mathfrak g$.  Hence the identity principle implies
that $(L_X+R_X)E(f)=0$ for all $X\in \mathfrak g$ and consequently
$E(f)$ is the desired radial extension of $f$.

\bigskip\noindent
For the necessity of the divisibility condition, we just reverse the
argument: If $E(f)$ is the radial extension of $f$, then 
$E(f)\circ \mathrm {exp}\vert U$ is annihilated by the
operators $L_X+R_X$ and the identity principle implies that
$E(f)\circ \mathrm {exp}$ is the radial extension of $\hat f$. 
Consequently $\hat f$
satisfies the divisibility condition which by definition is
the divisibility condition for $f$.
\end {proof}
\subsection {Jacobian formula}
Recall that we have regarded the Laplace-Casimir operators
as being differential operators $D:\mathcal R\to \mathcal R$
on the sheaf of radial holomorphic superfunctions on the Lie
supergroup associated to a Lie superalgebra $\mathfrak g$.
Here we restrict our considerations to the setting of
Theorem \ref{extension} so that the algebra of global radial functions
is decribed by the divisibility condition along a given maximal
torus $H$ in $G_0$. Note that divisibility at the group level means
that for every odd root $\beta $ the directional derivative $X_\beta (f)$
is divisible by $(r_\beta -1)$ where $r_\beta $ is the character
associated to $\beta $.

\bigskip\noindent
Denote by $\mathcal D_H$ the image in $\mathcal O(H)^W$
of the restriction map $R$ and $E:\mathcal D_H\to \mathcal R(G_0)$
the extension map which is $R^{-1}$. The associated \emph{radial part}
of a Laplace-Casimir operator $D$ on $H$ is defined as 
$\dot D:=RDE:\mathcal D_H\to \mathcal D_H$. Our goal in this paragraph
is to describe a basic result of Berezin which, at least for 
the Lie superalgebras and operators of main interest for the applications
in \cite{CFZ} and \cite{HPZ}, shows that the study of the
radial operators $\dot D$ can be reduced to analyzing certain
constant coefficient polynomial differential operators on $H$.
This is proved by introducing a sort of change of variables along
$H$ so that generically along $H$ one has a local product
decomposition in the $H$-direction and the transversal direction of the
supergroup action. As a consequence, a Jacobian $J$ appears and therefore
we refer to the result as the Jacobian formula.

\bigskip\noindent
The Jacobian $J$ is defined as follows as a meromorphic function on 
$H$: For $\zeta =\mathrm {exp}(t)\in H$    
\begin{displaymath}
    J(\zeta ) := \frac{\prod_{\alpha\in\Delta_0^+} 2\sin\frac{\alpha(
    t)}{2}} {\prod_{\beta\in\Delta_1^+} 2 \sin\frac{\beta(t)}
    {2}}
\end{displaymath}
where $\Delta^+ = \Delta_0^+ \cup \Delta_1^+$ is a system of even and
odd positive roots.  Under certain restrictive conditions the
Jacobian formula states that given a Laplace-Casimir operator $D$
there is a uniquely defined polynomial operator with constant 
coefficients $P_D$ on $H$ so that $\dot D=J^{-1}P_DJ$.  It should
be remarked that in the classical case of Lie groups the analogous
formula (without the odd roots in the Jacobian) holds in great 
generality.  In the Lie supergroup setting we state it in the
cases of $\mathfrak g=\mathfrak {gl}(m,n)$ and $\mathfrak {osp}(2m,2n)$.
In \cite{B} the latter Lie superalgebra is denoted by $C(m.n)$. From now on
$\mathfrak g$ is restricted to be one of these (complex) Lie superalgebras
equipped with the nongenerate bilinear form $(X,Y)=\mathrm {STr}(XY)$.

\bigskip\noindent
In the case of $\mathfrak {gl}(m,n)$ we choose $\mathfrak h$ to
be the Cartan algebra of diagonal matrices with coordinates
\begin {gather*}
h=
\begin {pmatrix}
\varphi & 0\\
0 & \psi 
\end {pmatrix}
.
\end {gather*}
One checks that the polynomial function 
$$
F_k=\mathrm {STr}(h^k)=\sum \varphi_i^k-\sum \psi_j^k
$$
satisfies the divisibility condition and therefore is extendible
to an $\mathrm {ad}$-invariant element of $S(\mathfrak g^*)$, i.e.,
to an element of $Z(\mathfrak g)$.  

\bigskip\noindent
In the case of $\mathfrak {osp}(2m,2n)$ we recall that 
$\mathfrak g_0=\mathfrak {so}_{2m}\oplus \mathfrak {sp}_{2n}$ and
as above choose $\mathfrak h$ to be in diagonal form with coordiantes
\begin {gather*}
h_{\mathfrak {so}}=
\begin {pmatrix}
\varphi & 0\\
0 & -\varphi 
\end {pmatrix}
\end {gather*}
and
\begin {gather*}
h_{\mathfrak {sp}}=
\begin {pmatrix}
\psi & 0\\
0 & -\psi 
\end {pmatrix}
\end {gather*}
with the full Cartan algebra given by 
$h=\mathrm {Diag}(h_{\mathfrak {so}},h_{\mathfrak {sp}})$. In this case
one defines the extendible polynomials
$$
F_k:=\frac{1}{2}\mathrm{STr}(h^{2k})=\sum\varphi_i^{2k}-\sum\psi_j^{2k}\,.
$$
In the case of $\mathfrak {gl}(m,n)$ the algebra of extendable polynomials
consists of polynomials with constant coefficients in the $F_k$.
In the case of $\mathfrak {osp}(2m,2n)$ one requires one additional
generator which is most conveniently chosen as 
$L=\varphi _1\cdot \ldots \cdot \varphi_mR$ where $R$ is the product
of the odd positive root functions.  For the sake of brevity of
notation we let
$$ 
\tilde F_k=F_k(\frac{1}{i}\frac{\partial}{\partial t})
$$
be the constant coefficient differential operator defined by $F_k$.
A simplified version of Berezin's Jacobian theorem 
(see Theorem 3.2 on p. 302 of \cite {B}) can be stated as follows. 
\begin {theorem} \label{jacobian formula}
Let $\mathfrak g$ be either $\mathfrak {gl}(m.n)$ or 
$\mathfrak {osp}(2m,2n)$ and $P$ be an extendible polynomial
on $\mathfrak h$ which defines the radial differential operator
$\dot D_P$ by $\dot D_P(f)=RD_PE(f)$ for $f\in \mathcal D_H$ a $W$-invariant
function satisfying the divisibility condition.  Then there exists a
uniquely determined polynomial function $T(P)$ on $\mathfrak h$ with associated
constant coefficient differential operator denoted by
$T(P)(\frac{1}{i}\frac{\partial}{\partial t})$ so that
$$
\dot D_P=J^{-1}T(P)(\frac{1}{i}\frac{\partial}{\partial t})J\,.
$$
Furthermore, if $P=F_k$, then $T(P)$ is of degree $k$ with top degree
term $F_k$.
\end {theorem}
It should be underlined that the partial derivative operators
$$
\frac{\partial}{\partial t}=(
\frac{\partial}{\partial t_1},\ldots ,\frac{\partial}{\partial t_{m+n}})
$$
are defined by the coordinates $(\varphi, \psi)$ of $\mathfrak h$ 
which were introduced above.

\bigskip\noindent
{\bf Remark.} For applications in \cite{HPZ} the first author was
originally interested in a local version of Theorem \ref{jacobian formula}
at a generic point of a maximal torus.  Using techniques which are
much closer to the methods of Helgason in the classical case, the
second author has proved such a result in his thesis \cite{Ka1}.
As is shown in the sequel, Berezin's global result 
Theorem \ref{jacobian formula} implies the local result.
Vice versa, application of the identity principle shows that
the local result implies Berezin's global result.  Thus the two
results are equivalent. Nevertheless we feel that it is of 
interest to have new viewpoint on these matters.  Hence, the
statement and a sketch proof of the local theorem have been
included as an appendix to this paper.

\bigskip\noindent
Now we turn 
to understanding the mapping $P\mapsto T(P)$.  This will be discussed
for the polynomials $P=F_k$ in both cases $\mathfrak {gl}(m,n)$ and
$\mathfrak {osp}(2m,2n)$. 
\subsection {Finite-dimensional representations}
Here we explain how to compute the polynomials
$T(P)$ in terms of the eigenvalues of the radial operators
$\dot D_{P}$ on characters of \emph{finite-dimensional} 
representations. We restrict to the cases 
$\mathfrak g=\mathfrak {gl}(m,n),\mathfrak {osp}(2m,2n)$, but most
of the discussion applies in a much more general setting, e.g.,
where $\mathfrak {g}_0$ is semisimple. For more details see 
Chapter 3.10 (p.307-311) of \cite {B}.
\subsubsection* {Character formula}
Consider a finite-dimensional irreducible representation $\rho $ of a
complex Lie supergroup associated to one of the Lie
superalgebras $\mathfrak g=\mathfrak {gl}(m,n),\mathfrak {osp}(2m,2n)$. 
This is by definition a homomorphism of the Lie supergroup associated
to $\mathfrak g$ to that associated to Lie superalgebra $\mathfrak {gl}(V)$.
Such is defined by a holomorphic mapping $G_0\to \mathrm {GL}(V)$
which lifts to the sheaf level as a mapping 
$\mathcal {F}_{\mathrm {GL}(V)}\to \mathcal {F}_{G_0}$ which preserves
the defining Lie supergroup triples (see \cite {B}, p.248). Taking
a basis of homogeneous elements of $V$ one interprets 
$\rho $ as a holomorphic map of $G_0$ to matrices whose entries 
are superfunctions.  The character of such a representation is defined by
$\chi (g):=\mathrm {STr}(\rho(g))$. It is a radial superfunction on $G_0$
and  we consider its restriction $\chi $ to a Cartan subgroup $H$.  
It is an eigenfunction of every radial operator $\dot D$.
In other words there is a 
homomorphism $\lambda $ defined on the space of radial differential
operators with values in $\mathbb C$ so that 
$\dot D(\chi )=\lambda (\dot D)\chi $.

\bigskip\noindent
Now apply the Jacobian formula, 
$\dot D_P=JT(P)(\frac{1}{i}\frac{\partial}{\partial t})J^{-1}$,
define $\tilde \chi =J\chi $ and observe that
$$
T(P)(\frac{1}{i}\frac{\partial}{\partial t})\tilde \chi =
\lambda (\dot D_P)\tilde \chi\,.
$$
In other words, the eigenvalue homomorphism for the radial operator
$\dot D_P$ on the character $\chi $ is the same as the eigenvalue
homomorphism for the constant coefficient operator   
$T(P)(\frac{1}{i}\frac{\partial}{\partial t})$ on the function
$\tilde \chi $.  This simple remark leads to an exact description
of $T(P)$ in terms of eigenvalues of irreducible representations.

\bigskip\noindent
For this note that $J$ is defined on
$H$ so that $\tilde \chi $ can be expanded in a Fourier series
$$
\tilde \chi (t)=\sum_{k\in \mathbb Z^{m+n}}a_ke^{i\langle k,t\rangle}\,.
$$
Applying $T(P)(\frac{1}{i}\frac{\partial}{\partial t})$ to both sides
one shows that if $a_k\not=0$, then 
\begin {equation}\label{eigenvalue}
T(P)(k)=\lambda (\dot D_P)
\end {equation}
for \emph{every} extendible polynomial $P$. Letting $P$ range
over all such polynomials one proves the following fact.
\begin {proposition}
The set of lattice elements $k$ such that $a_k\not=0$ is 
a $W$-orbit $W.k_0$.
\end {proposition} 
It should be noted that since $T(P)$ is itself $W$-invariant
the lack of uniqueness of the lattice element $k$ is minimal.

\bigskip\noindent
Now $\chi $ is $W$-invariant.  Furthermore, for $\sigma $ in the Weyl group
it follows that $\sigma (J)=\varepsilon(\sigma )J$ where 
$\varepsilon (\sigma )=\mathrm {det}(\sigma )=\pm 1$.   Hence,
up to a multiplicative constant
$$
\chi (t)=J^{-1}(t)\sum_{\sigma \in W}\varepsilon (\sigma)
e^{i\langle k_0,\sigma (t)\rangle}
$$
for any fixed $k_0$ in the support of $\tilde \chi $. 

\bigskip\noindent
Now order the weight lattice so that the roots $\alpha $ and $\beta $
which occur in the above products are positive and write 
$$
\chi (t)=\sum c_je^{i\langle m_j, t \rangle}
$$
where the $m_j$ are the weights of the representation $\rho $ with
$\Lambda $ being the highest weight which occurs. Compare this
expression for $\chi (t)$ with that above to obtain 
$$
\sum c_je^{i\langle m_j, t \rangle}
\prod (e^{i\frac{\alpha(t)}{2}}-e^{-i\frac{\alpha (t)}{2}})
=\sum \varepsilon (\sigma)e^{i\langle k_0,\sigma (t)\rangle}
\prod (e^{i\frac{\beta(t)}{2}}-e^{-i\frac{\beta (t)}{2}})\,.
$$ 
Equating the highest order terms on each side yields
$$
\Lambda +\frac{1}{2}\sum \alpha =k_0+\frac{1}{2}\sum \beta
$$
where $k_0$ is the highest of the elements in the $W$-orbit
$W.k_0$.  Turning this around, we see that
$$
k_0=\Lambda +\delta 
$$
where 
$$
\delta =\frac{1}{2}\big(\sum \alpha -\sum \beta \big)\,.
$$
\begin {theorem}
Let $\rho $ be a finite-dimensional representation of a 
Lie supergroup associated to one of the Lie superalgebras
$\mathfrak g=\mathfrak {gl}(m,n),\mathfrak {osp}(2m,2n)$.
Let $P$ be an extendible polynomial on $\mathfrak h$ and
$T(P)$ be the polynomial which is defined by
$$
\dot D_P=J^{-1}T(P)(\frac{1}{i}\frac{\partial}{\partial t})J\,.
$$
If $\chi $ is the character of $\rho $ with
the homomorphism $\lambda $ defined by 
$$
\dot D_P=\lambda (D_P)\chi\,,
$$
then
$$
\lambda (\dot D_P)=T(P)(\Lambda +\delta)
$$
where $\Lambda $ is the highest weight of $\rho $.
\end {theorem}
\begin {proof}
This follows immediately from (\ref{eigenvalue}) and the fact
that our choice of $k=k_0$ in the $W$-orbit is $k_0=\Lambda +\delta$.
\end {proof} 
\subsection {Generating functions}
Here we fix $\mathfrak {g}$ as one of the Lie superalgebras
$\mathfrak {gl}(m,n)$ or $\mathfrak {osp}(2m,2n)$ and let
$\dot D_\ell$ be the radial operator defined by the particular
extendible polynomial $F_\ell$. Using the Fourier series development
of characters of representation, it was shown above that the
value $\lambda (\dot D_\ell)$ of the eigenvalue homomorphism
on $\dot D_\ell$ for the character of an irreducible representation
$\rho $ of highest weight $\Lambda $ of the associated Lie supergroup
is the value of the polynomial $T(P)$ on lattice point
$k_0=\Lambda +\delta $. Letting $\rho $ range through all such
representations, we see that $T(P)$ is the unique polynomial with
this property. If we think of such a point $k_0$ as
a weight, then it is in $\mathfrak h^*$; so we reformulate the
result as follows: There is a uniquely determined polynomial
function $R_\ell$ on $\mathfrak h^*$ with 
$R_\ell(\Lambda +\delta )=\lambda (\dot D_\ell)$ on every irreducible
representation of highest weight $\Lambda $. 

\bigskip\noindent
Associated to the sequence $\{R_\ell\}$ of polynomials one has the
\emph{generating function}
$$
S(z):=\sum z^\ell R_\ell (x)
$$
which is computed in closed form in \cite{B} (see Lemma 4.3,  pages 327-329, 
for the case of $\mathfrak {gl}(m,n)$ and Lemma 4.4, pages 335-341,
for $\mathfrak {osp}(2m,2n)$).  The resulting formulas for the
polynomials $R_\ell$ are derived after the proofs of these lemmas.
Using the identificatiion $R_\ell=T(F_\ell)$, one has the following consequence
which we formulate simultaneously for both $\mathfrak {gl}(m,n)$
and $\mathfrak {osp}(2m,2n)$.
\begin {theorem}\label {basic formula theorem} 
If $\dot D_\ell$ is the radial differential operator defined by
the extendible polynomial $F_\ell$ with
\begin {equation}\label{basic formula}
\dot D_\ell=J^{-1}T(F_\ell)(\frac{1}{i}\frac{\partial}{\partial t})J\,,
\end {equation}
then $T(F_\ell)\in \mathbb C[F_1,\ldots ,F_\ell]$. Moreover
$T(F_\ell)=F_\ell+Q$ where $Q\in \mathbb C[F_1,\ldots ,F_{\ell-1}]$ 
is a polynomial of lower degree.
\end {theorem}

\subsection {An application}
In \cite {CFZ} and \cite {HPZ} characters $\chi $ of  
representations Lie supergroups on certain infinite-dimensional 
spaces play an important role. In \cite {CFZ} the complex
Lie superalgebra at hand is $\mathfrak {gl}(m,n)$ and in
\cite {HPZ} it is $\mathfrak {osp}(2n,2n)$.  In these situations one 
would hope to apply the above results on radial operators. However,
this can not be directly done, because the characters
are defined by supertrace and only converge on certain
open domains $\mathcal H$ in $G_0$ or on finite covering spaces 
$\widehat {\mathcal H}$ of such domains.
On the other hand, Laplace-Casimir operators are local and
can therefore be applied to such characters and in the settings
of \cite{CFZ} and \cite{HPZ} the characters $\chi $ which appear 
are annihilated by Laplace-Casimir operators $D_\ell$ defined by the $F_\ell$.

\bigskip\noindent
In the domains $\mathcal H$ or the covering spaces $\widehat {\mathcal H}$
there are closed connected complex submanifolds $T^+$ 
which are either open subsets
of a Cartan algebra $H$ or lifts of such into the covering space.
Now the radial operators $\dot D_\ell$ are differential operators
which are apriori defined on the space $\mathcal D_H$ of globally
defined extendible $W$-invariant holomorphic functions and on
that space we know how to compute them using the righthand side
of (\ref{basic formula}).  The restrictions of the
characters $\chi $ to $T^+$, which are by definition the numerical
parts of $\chi \vert T^+$, are by definition extendible 
as radial superfunctions, but they are only defined on $T^+$ and not on $H$.  
Nevertheless we wish to show that they are annihlated by the 
operators which are described by the righthand side of (\ref {basic formula}).
For this we prove a local version of (\ref {basic formula})
and obtain the desired result on $T^+$ by applying the identity principle.

\subsubsection* {Local formula for $\mathbf{\dot D_\ell}$}
We refer to a point in $H$ as being \emph{superregular} if
it is regular in the sense of Lie theory and is not contained
in any of the odd root hypersurface $\{r_\beta=1\}$. Every  
a superregular point $x$ has a basis of open neighborhoods $V$
in $\mathfrak H$ which are relatively compact in the set of
superregular points in $H$ with the property that 
$\sigma (V)\cap V=\emptyset $ for every
$\sigma \in W\setminus \{\mathrm {Id}\}$.
Given such a $V$ we thicken it
as follows to an open neighborhood $U$ in $G_0$.  Let $\Delta $
be a polydisk in $\mathfrak {g}_0$ which is transversal to 
$\mathfrak h$ and define $U=\{\mathrm {exp}(\xi).x; \xi \in \Delta, x\in V\}$.
We choose $\Delta $ small enough so that $U\cong \Delta \times V$. For
$x\in V$ fixed we think of $\mathrm {exp}(\Delta).x$ as a local orbit
of $G_0$.
%
%\begin {proposition}
%Every superregular point $x\in H$ has a basis of open neighborhoods
%$U$ in $G_0$ so that $U\cap \mathfrak h$ consists only of superregular
%points and such that the restriction map $\mathcal O(G_0)\to \mathcal O(U)$
%has dense image. Furthermore, for $\Delta $ a sufficiently small
%polydisk transversal to $\mathfrak h$ in $\mathfrak {g}_0$, the
%neighborhoods $U$ can be chosen as 
%$U=\mathrm {exp}(\Delta ).(U\cap H)$.  Also need $\sigma (U)\cap U=\emptyset$
%for all $\sigma \in W$.
%\end {proposition}
%\begin {proof}
%The first property follows from the fact that the set of superregular
%points in open. The second property is a fundamental (standard) result from
%complex analysis.
%\end {proof}
%
%Now let $E(f)$ be a radial holomorphic function on $U$ where $U$
%is as in the above proposition. It is completely determined
%as a holomorphic map $f:U\to \wedge \mathfrak {g}_1^*$ by 
%its restriction to $U\cap H$. Furthermore, since $E(f)$ is radial
%the numerical part $f$ of this restriction completely determines
%$E(f)$. The converse statement also holds. 
%
\begin {proposition}
Every superregular point $x$ in $H$ has a neighborhood basis
of open sets $V$ and $U$ as above so that the restriction map 
$R:\mathcal R(U)\to \mathcal O(V)$
is an isomorphism.
\end {proposition}
\begin {proof}
Since holomorphic maps $U\to \wedge \mathfrak {g}_1^*$ which 
are invariant by the local conjugation-action of $G_0$
are completely determined by their restrictions to $V$, it
follows that $R$ is injective. Surjectivity is proved by
the following approximation argument. 

\bigskip\noindent
First, in order to take care of $W$-invariance we consider
the quotient $\pi:H\to Z=H/W$. The restriction $\pi\vert V $
maps $V$ biholomorphically onto a domain $\tilde V$.  A basic
theorem of complex analysis states that we may choose $\tilde V$
(and accordingly $V$) so that the restriction map 
$\mathcal O(Z)\to \mathcal O(\tilde V)$ has dense image.

\bigskip\noindent
Now let $R$ be the product of the odd root functions on $H$ and $\tilde R$
be the associated function on $Z$. Define $\tilde f$ be the function
on $\tilde V$ associated to a given holomorphic function $f$ on $V$.
Let $\tilde f_n$ be a sequence of holomorphic functions on $Z$ which
converge to $\tilde R^{-2}\tilde f$ in $\mathcal O(\tilde V)$.  It
follows that $\tilde h_n:=\tilde R^2\tilde f_n$ converges to $\tilde f$. 
The point of this construction is that the sequence $\{h_n\}$ of lifts
defined by $h_n:=\pi^*(h_n)$ converge to $f$ on $V$. In addition
these are $W$-invariant and have the divisibility property. Thus we
have the sequence $\{E(h_n)\}$ of radial extensions. Now
the extension $h_n\vert V\to E(h_n)\vert V$ is such that 
the convergence of $h_n\vert V$ implies the convergence of 
$E(h_n)\vert V$ as a sequence of $\wedge \mathfrak {g}_1^*$-valued
holomorphic maps.  Consequently the maps $E(h_n)\vert U$ which
are constant along the local $G_0$-orbits defined by $\Delta $
also converge.  If a sequence of holomorphic functions converges,
then so does any induced sequence of derivatives.  Thus the
limit $E(f)$ of the sequence $\{E(h_n)\vert U\}$ is a radial 
holomorphic function whose (numerical) restriction to $V$ is
the given function $f$.
\end {proof}
Having localized to the open sets $V=U\cap H$ of superregular elements
of $H$ and proved the above extension result, given a Laplace-Casimir
operator $D$ we define its radial part on $U\cap H$ in the same
way as in the global case: $\dot D_{U\cap H}(f):=RDE(f)$.  Since
the extension result was proved by taking limits of globally
defined extendible functions and
the global operator $\dot D$ is continuous, it
follows that $\dot D_{U\cap H}$ is just the restriction 
of $\dot D$ to $U\cap H$. Thus we have the following local
version of Theorem \ref{jacobian formula}. 
\begin {theorem}\label {jacobian theorem local}
Under the assumptions of Theorem \ref{jacobian formula},
let $D$ be a Laplace-Casimir operator defined by an extendible
polynomial $P$. If $U$ is as above, then the domain of definition of
the radial operator $\dot D_{U\cap H}$ is the full algebra
of holomorphic functions $\mathcal O(U\cap H)$ and 
$$
\dot D_{U\cap H}=
J^{-1}T(P)(\frac{1}{i}\frac{\partial}{\partial t})J\,.
$$
\end {theorem}
As a result we have the local version of Theorem \ref{basic formula theorem}
\begin {theorem}\label {local basic formula}
Under the assumptions of Theorem \ref{basic formula theorem},
for $U$ as above it follows that
$$
\dot D_{U\cap H}=J^{-1}(\tilde F_\ell +Q(\tilde F_1,\ldots \tilde F_{\ell-1}))J
$$
where $Q$ is a polynomial operator of lower degree than $\tilde F_\ell$.
\end {theorem}
Now let us return to the settings of \cite {CFZ} and
\cite {HPZ} where we have an open piece $T^+$ of the Cartan
algebra $H$ contained as a closed submanifold of a domain $\mathcal H$
in $G_0$ or a finite-to-one covering space $\widehat {\mathcal H}$
of such a domain.  The Weyl group $W$ acts on these domains so that
the above arguments apply: If $x\in T^+$ is superregular, then we
setup $U$ as above and prove the following result.  
\begin {theorem}
Let $D_\ell$ be the Laplace-Casimir operator defined on 
$\mathcal H$ or $\widehat {\mathcal H}$ by the extendible
polynomial $F_\ell$. Then a holomorphi superfunction  
on such a domain is annihilated by $D_\ell$ if and only
if its (numerical) restriction to $T^+$ is annhilated by
$J^{-1}(\tilde F_\ell +Q(\tilde F_1,\ldots \tilde F_{\ell-1}))J$.
\end {theorem}
\begin {proof}
Since $D_\ell$ acts on the full sheaf of radial functions,
we may regard it as acting on the restriction of given holomorphic
superfunction to $U$. Thus its restriction to $U\cap H$
is annihilated by the associated radial operator $\dot D_{U\cap H}$.
By Theorem \ref{local basic formula}
the operator $J^{-1}(F_\ell +Q(F_1,\ldots F_{\ell-1})J$
annihilates this restriction on $U\cap H$ and the desired result
follows by the identity principle.
\end {proof}
\subsection* {Appendix: local representation of  radial operators}
Here, referring to \cite{Ka1,Ka2} for details, we outline a proof
of a local version of Theorem \ref{jacobian formula}.  We do this
in the setting of complex Lie supergroups. Let us begin
with the preparation which is required to state the result.

\bigskip\noindent
Recall that if $G$ is a reductive complex Lie group with maximal
complex torus $H$, then $h\in H$ is by definition regular if
its centralizer in $G$ is just $H$ itself.  The set of regular
elements is open and dense in $H$ and every $h_0\in H_{reg}$ has
an open neighborhood $A$ in $H$ such that 
$G.A=\{ghg^{-1}; h\in A,\ g\in G\}$ is $G$-equivariantly a product
$A\times G/H$.  A local version of this which can be applied, 
e.g., in situations where $G$ is only acting locally on a neighborhood
of $h_0$ in $G$, can be stated as follows: 
There is an open neighborhood $B$ of the identity in
the submanifold of $G$ which is defined to be the product of all
$H$-root spaces so that $B.A$ is an open neighborhood of $h_0$
which is (locally) $B$-equivariantly the product $A\times B$.

\bigskip\noindent
It is shown in \cite{Ka2} that this local product decomposition holds 
for Lie supergroups $(G,\mathcal F)$ of type I which are 
equipped with a nondegenerate invariant bilinear form 
$b:\mathrm {Der}(\mathcal F)\times \mathrm {Der}(\mathcal F)\to \mathcal F$.
Recall that ``type I'' means that Cartan algebras are even, i.e., are
contained in $\mathfrak {g}_0$ and all root spaces are 1-dimensional.
The complex manifolds $A$ and $B$ are defined exactly as above.
The complex subsupermanifold structure on $A$ is even, but nevertheless
we denote it by $\mathcal A$ to emphasize that it is a subsupermanifold.
The subsupermanifold structure on $B$ is not even.  It is analogously
denoted by $\mathcal B$.
\begin {theorem}
For $h_0$ a regular element of $H$ there is an open (supermanifold)
neighborhood $\mathcal {B}.\mathcal {A}$ which is locally 
$\mathcal {B}$-equivariantly isomorphic to
the product $\mathcal {A}\times \mathcal {B}$.
\end {theorem}
In this context a superfunction is radial if it is annihilated by
all derivations in $\mathrm {Der}(\mathcal B)$ which, using the
product structure at hand, is regarded as being contained in
the space of derivations on the product neighborhood 
$\mathcal A\times \mathcal B$.  A (super) differential operator $D$
on this neighborhood is said to be radial if it maps
radial functions to radial functions.  

\bigskip\noindent
Using the product structure, one extends a holomorphic superfunction on
$\mathcal A$, which is by definition a numerical functions,
to a unique holomorphic radial superfunction on $\mathcal A\times \mathcal B$.
Applying a radial operator $D$ to
such a function, we again obtain a radial function which is
uniquely determined by a function on $\mathcal A$. Thus
we obtain a classical differential operator 
$\dot D :\mathcal O(A)\to \mathcal O(A)$ on numerical functions.

\bigskip\noindent
The main source of radial operators is the center of the universal
enveloping algebra.  The resulting operators
on $\mathcal A\times \mathcal B$ are called Laplace-Casimir 
operators.  The local version of Theorem \ref{jacobian formula} which
is proved in detail in \cite{Ka1}, Chapter 4, is formulated below.  As in
Theorem \ref{jacobian formula} this requires not only type I
and the existence of the supersymmetric invariant form, but
also that the function $J$ is an eigenfunction of a second order
Laplacian which is constructed in the proof. Thus, to simplify
the formulation, we have only stated it for the cases of relevance
for the applications in \cite{HPZ} and \cite{CFZ} where these assumptions
are satisfied.
\begin {theorem}\label{local radial}
If $\mathfrak g$ is either $\mathfrak {gl}(m,n)$ or
$\mathfrak {osp}(2m,2n)$, then for every 
Laplace-Casimir operator $D$ on 
$\mathcal A\times \mathcal B$ there is a uniquely determined 
polynomial constant coefficient operator $P$ on $\mathcal A$ so
that $\dot D=J^{-1}PJ$.
\end {theorem}
Here $J$ is the same globally defined function as
in Theorem \ref{jacobian formula}.  

\bigskip\noindent
Let us conclude this appendix by commenting on the proof
of this local result.  As in the classical case the main
point is to define a Laplace operator which in end effect
is a second order Laplace-Casimir operator and whose
radial part is (by an explicit computation) of the form
in the theorem. Using the invariant form $b(x,y)=\mathrm {STr}(xy)$, 
this operator is constructed as follows.

\bigskip\noindent
Fixing a splitting, one obtains a deRham complex
of holomoprhic differential forms in $\wedge (T^*_0M\oplus \Pi T^*_1M)$,
where here $M$ is the supermanifold $\mathcal A\times \mathcal B$ of dimension
$k\vert \ell $.  The operation $\Pi$ exchanges the even and odd parts
of a graded vector space.  From the point of view of integration theory the
correct ``top-dimensional'' form is a nowhere vanishing section
$\omega _M°$ (a ``Berezinian'') of $\wedge ^{k+\ell}(T^*_0M\oplus \Pi T_1M)$.
Using $b$ this is transformed to a top-dimensional holomorphic
``volume form'' $\omega _M$ in $\wedge ^{k+\ell}(T^*_0M\oplus \Pi T^*_1M)$.
Again using $b$ to define contraction of $\omega _M$ with 1-forms
to obtain a $\ast $-operation so that $\partial \ast \partial f$ 
is a multiple $h\omega_M$ and then defining $\ast $ from top-dimensional
forms back to functions, one obtains a Laplace operator 
$L_M:=\ast \partial \ast \partial $ on superfunctions 
(see \cite{Ka1}, 4.1.1-4.1.3). 

\bigskip\noindent
In order to define $L_M$ one only needs a Lie supergroup of type I
with a nondegenerate supersymmetric bilinear form.  It is a radial
operator and a formula for its radial part can be explicitly 
computed (\cite{Ka1}, Lemma 4.4).  In order to obtain a formula
of the type in Theorem \ref{local radial} the following nontrivial
result is needed. Here $J$ is the function defined above which
is qualitatively described as the square-root of the superdeterminant
of the Jacobian of the mapping which identifies $\mathcal A\times \mathcal B$
with an open subset of the Lie supergroup $(G,\mathcal F)$.
\begin {theorem} 
If $\mathfrak g$ is either $\mathfrak {osp}(2m,2n)$ or
$\mathfrak {gl}(m,n)$, then $J$ is an eigenfunction of $L_M$.
\end {theorem}
One way of seeing this is to first show that $L_M$ is (up to a constant
multiple) in fact the operator defined by the second order element of the
center of the universal enveloping algebra which is defined in the usual way
by $b$ (\cite{Ka1}, Lemmas 4.7 and 4.8). Then one computes explicitly
to show that $J$ is indeed an eigenfunction.  In general such computations
are quite involved (\cite {B} Theorems 4.1 and 4.4), 
but in certain cases which are relevant for applications they are simple 
(see \cite{HPZ}, $\S4.3$) 

\bigskip\noindent
Due to the fact that $J$ is an eigenfunction one can express
the radial part of the second order Laplacian as in Theorem
\ref{local radial}.  The general theorem follows immediately
in the classical way using in addition the fact that 
the given Laplace-Casimir operator commutes with $L_M$ 
(\cite{Ka1}, Theorem 4.1).   
\parbox[t]{6cm}{
Alan Huckleberry\\
Fakult\"at f\"ur Mathematik\\
Ruhr-Universit\"at Bochum,\\ 
Universit\"atsstra\ss e 150\\ 
D-44801 Bochum, Germany, and\\ 
School of science and engineering\\ 
Jacobs University\\
Bremen Compus Ring 1,\\
D-28759 Bremen, Germany\\
ahuck@cplx.rub.de}
\parbox[t]{6cm}{
Matthias Kalus\\
Fakult\"at f\"ur Mathematik\\ 
Ruhr-Universit\"at Bochum,\\ 
Universit\"atsstra\ss e 150\\ 
D-44801 Bochum, Germany\\
Matthias.Kalus@rub.de}

\begin {thebibliography} {XXX}
\bibitem [B] {B}
Berezin, F.A.: Introduction to superanalysis, D. Reidel Publishing
company, 1987
\bibitem [CFZ] {CFZ}
Conrey, J.B., Farmer, D.W., Zirnbauer, M.R.: Howe
pairs, supersymmetry and ratios of random characteristic polynomials
for the unitary group (math-ph/0511024)
\bibitem [HPZ] {HPZ}
Huckleberry,~A., P\"uttmann,~A. and Zirnbauer,~M.:
Haar--expectations of ratios of random characteristic
polynomials (arxiv:0709.1215)
\bibitem [Ka1] {Ka1}
Kalus, M.:
Complex analytic aspects of Lie supergroups,
Dissertation of the Ruhr-Universit\"at Bochum (90 pages ms., 
submitted)
\bibitem [Ka2] {Ka2}
M. Kalus, Almost complex structures on real Lie 
supergroups (16 pages ms, submitted, arXiv:1012.4429)
\bibitem [K] {K}
Kostant, B: Graded manifolds, graded Lie theory and prequantization,
Lecture Notes in Mathematics 570, Springer Verlag, 1987, 177-306
\bibitem [Ko] {Ko}
Koszul, J. L.: Graded manifolds and graded Lie algebras, Proceeding of
International Meeting on Geometry and Physics (Bologna) Pitagora, 1982,
71-84
\bibitem [S] {S}
Sergeev, A.: The invariant polynomials on simple Lie superalgebras,
Representation Theory, {\bf 3} 250-280 (19999 
\bibitem [V] {V}
Vishnyakova, E. G.: On complex Lie supergroups and split homogeneous
supermanifolds (arXiv: 0811.2581)
\end {thebibliography}
\end {document}